\documentclass[sigconf]{acmart}

\renewcommand\footnotetextcopyrightpermission[1]{} 
\AtBeginDocument{%
	\providecommand\BibTeX{{%
			\normalfont B\kern-0.5em{\scshape i\kern-0.25em b}\kern-0.8em\TeX}}}

\copyrightyear{2022}
\acmYear{2022}
\setcopyright{acmcopyright}\acmConference[ICCAD '22]{IEEE/ACM International Conference on Computer-Aided Design}{October 30 - November 3, 2022}{San Diego, California, USA}
\acmBooktitle{IEEE/ACM International Conference on Computer-Aided Design (ICCAD '22), October 30 - November 3, 2022, San Diego, California, USA}
\acmPrice{15.00}
\acmDOI{10.1145/nnnnnnn.nnnnnnn}
\acmISBN{978-x-xxxx-xxxx-x/YY/MM}

\usepackage{nicematrix}
\usepackage{array}
\usepackage{color}
\makeatletter  
\newif\if@restonecol  
\makeatother

\usepackage[linesnumbered,ruled,vlined]{algorithm2e}
	\usepackage{algpseudocode}  
	\usepackage{amsmath}  
	\usepackage{graphicx}
	\usepackage{arydshln}
	\usepackage{multirow}
	\usepackage{siunitx}
	\usepackage{makecell}
	\usepackage{subfigure}

	\newcommand{\quanState}[1][]{|#1\rangle}
	\newcommand{\conjTran}[1][]{{#1}^{\dagger}}
	\newcommand{\inverse}[1][]{{#1}^{-1}}
	\newcommand{\varSet}{\mathbf{X}}
	\newcommand{\varSetX}{\mathbf{X}}
	\newcommand{\varSetY}{\mathbf{Y}}

	\newcommand{\order}{\pi}
	\newcommand{\size}[1][]{|#1|}
	\newcommand{\set}[1]{\{ #1 \}}
	\newcommand{\ind}[1][]{i(#1)}
	\newcommand{\ie}{{\it i.e.}}
	
	\newcommand{\bool}{\mathbf{B}}

\begin{document}
		
		\title{Quantum Multiple-Valued Decision Diagrams with \\ Linear Transformations}
		
		
	
	
	\author{Yonghong Li}
	\affiliation{%
		 \institution{College of Information Science and Technology}
		 \city{Jinan University}
		 \country{China}}
	\email{lyh2020@stu2020.jnu.edu.cn}

	\author{Miao Hao}
	\affiliation{%
		 \institution{College of Information Science and Technology}
		\city{Jinan University}
		\country{China}}
	\email{miaohao77@stu2020.jnu.edu.cn}
	
	
	\begin{abstract}
		Due to the rapid development of quantum computing, the compact representation of quantum operations based on decision diagrams has been received more and more attraction.
		Since variable orders have a significant impact on the size of the decision diagram, identifying a good variable order is of paramount importance.
		In this paper, we integrate linear transformations into an efficient and canonical form of quantum computing: Quantum Multiple-Valued Decision Diagrams (QMDDs) and develop a novel canonical representation, namely linearly transformed QMDDs (LTQMDDs). 
		We design a linear sifting algorithm for LTQMDDs that search a good linear transformation to obtain a more compact form of quantum function.
		Experimental results show that the linear sifting algorithm is able to generate decision diagrams that are significantly improved compared with the original sifting algorithm.
		Moreover, for certain types of circuits, linear sifting algorithm have good performance whereas sifting algorithm does not decrease the size of QMDDs.

	\end{abstract}


	
	
	\maketitle
	
	\section{Introduction}
	\looseness=-1
	Quantum computing is one of computation models that utilizes the properties of quantum mechanics to solve computational problems.
	Due to the parallelism of quantum computing, quantum computers are capable of solving some specific computation problems (\ie, integer factorization \cite{Shor1999}, database search \cite{Gro1996}, computational biology \cite{FedG2021} and quantum chemistry \cite{ReiM2017,ArgJ2019}) substantially faster than classical computers.
	A quantum state over $n$ qubits is formalized by a normalized vector of size $2^n$.
	The state over $n$ qubits can be transformed via a quantum operation represented as a unitary matrix of size $2^n \times 2^n$.
	As more qubits involve, the size of the normalized vector and unitary matrix for quantum computing grows exponentially.
	
	
	
	
	\looseness=-1
	To mitigate inefficiency in representation, many different approaches are proposed, for example, based on arrays \cite{GueG2020, JonT2019, GheV2018}, tensor networks \cite{MarL2008, WanS2017}, and decision diagrams (DDs) \cite{WanLT2008, NieWM2016}.
	In this paper, we focus on the compact DD-based forms. 
	The shared insight behind DDs is to recursively decompose the unitary matrix into submatrices according to a variable order.
	The choice of variable order has a significant impact on the size of decision diagrams.
	Therefore, identifying a good variable order for DDs is of paramount importance.
	\citet{NieWM2016} propose a canonical DD-based form of quantum functionality, namely Quantum Multiple-Valued Decision Diagrams (QMDDs).
	They also developed a method to exchange adjacent variables in QMDDs, and integrated sifting algorithm \cite{Rud1993} into QMDDs that relies on the exchange method.
	Compared to the approach without variable reordering, sifting algorithm produces much more compact QMDDs.
	
	\looseness=-1
	From the mathematical point of view, each variable order is an automorphism on $n$ Boolean variables $\bool^n$, \ie, a bijection $\bool^n \to \bool^n$.
	Based on an automorphism, we can obtain a different matrix $M'$ via relocating the position of each entry of the original matrix $M$.
	The new matrix $M'$ may have a smaller QMDD-representation than the original one $M$.
	Linear transformation, which is a fully ranked order of linear combination of variables, is a more expressive representation of automorphisms than variable orders.
	It was confirmed that binary decision diagrams (BDDs \cite{Bry1986}) with linear transformations have smaller sizes than those with only variable orders from both perspectives of theory and practice \cite{MeiST2000,GunD2000,GunD2003}.

	Inspired by the concept of linear transformation, in this paper, we design linear sifting algorithm for QMDDs so as to acquire a more compact form of quantum functionality.
	To this end, we first give a new definition of linear transformation and show how linear transformations change unitary matrices.
	Then, we incorporate linear transformations into QMDDs and derive a new representation, namely linearly transformed QMDDs (LTQMDDs).	
	In fact, a LTQMDD denotes the changed unitary matrix based on the given linear transformation.
	Moreover, we devise three level exchange procedures for swapping the nodes of two adjacent levels in LTQMDDs and develop linear sifting algorithm that searches for a good linear transformation for LTQMDDs based on the exchange procedures.
	Finally, we implement linear sifting algorithm for LTQMDDs and compare linear sifting algorithm with original sifting algorithm.
	The empirical results show that linear sifting algorithm is able to generate LTQMDDs of $11\%$ smaller size than original sifting algorithm.
	%

	
	\looseness=-1
	The structure of this paper is organized as follows.
	Some essential concepts of quantum computing, QMDDs and linear transformations are briefly reviewed in Section 2. 
	Section 3 introduces the integration of linear transformations into quantum computing and QMDDs and illustrates a minimization algorithm for LTQMDDs.
	Experimental results are presented in Section 4 followed.
	Finally, Section 5 concludes this paper.
	
	\section{Preliminaries}	
	\looseness=-1
	This section includes the basic knowledge of quantum computing, quantum multiple-valued decision diagrams (QMDDs), and linear transformations.
	
	\subsection{Quantum Computing}
	\looseness=-1
	Throughout this paper, we fix a set of $n$ variables $\varSet: \set{x_0, x_1, \cdots, \\ x_{n - 1}}$ and $\bool = \set{0, 1}$.
	In quantum computing, the elementary unit of quantum information is \textit{quantum bits} (in short, \textit{qubits}). 
	$n$ qubits form an $n$-level quantum system, which is formalized by a $2^n$-dimensional Hilbert space over complex numbers.
	The orthonormal basis states consists of $2^n$ states of which is represented by $\quanState[v]$ where $v \in \bool^n$.
	A \textit{quantum state} of $n$ qubits $\quanState[\varSet]$ is a linear superposition of orthonormal basis states $\sum_{v \in \bool^n} (\alpha_v \cdot \quanState[v])$ where each $\alpha_v$ is the coefficient of $\quanState[v]$ and $\sum_{v \in \bool^n} |\alpha_v|^{2} = 1$.
	
	\looseness=-1
	A natural number $i$ can be represented in a $n$-dimensional boolean vector and vice versa.
	For example, the $4$-dimensional boolean vector of $5$ is $(0101)_2$.
	For ease of presentation, we use these two representations interchangeably.
	A \textit{quantum state} of $n$ qubits can be represented by a normalized vector of length $2^n$ where its $i$-th element\footnote{Throughout this paper, we use $0$-based indexing for vectors and matrices, \ie, the index of element of a vector starts from $0$.} is the coefficient of the basis state $\quanState[i]$.
	
		\begin{example}
				Suppose that we have 2 qubits and the set of variable $\varSet = \{x_{0}, x_{1}\} $. 
				The quantum state $\quanState[\varSet]$ is $\alpha_{0} \cdot \quanState[0] + \alpha_{1} \cdot \quanState[1] + \alpha_{2} \cdot \quanState[2] + \alpha_{3} \cdot \quanState[3]$.
				The vector of $\quanState[\varSet]$ is $[\alpha_{0} \ \alpha_{1} \ \alpha_{2} \ \alpha_{3}]^{\top}$.
			\end{example}
	
	A quantum state of $n$ qubits is transformed by a quantum operation, which is described as a unitary matrix of size $2^n \times 2^n$.
	A complex-valued matrix $U$ is \textit{unitary} iff its conjugate transpose $\conjTran[U]$ is its inverse $\inverse[U]$.
	
%
%
%

	\subsection{Quantum Multiple-Valued Decision Diagrams}
		\begin{figure*}[t]
	\subfigure[The matrix $U$ with the standard order {$[x_{0}, x_{1}]$} ]{
		\begin{minipage}[t]{0.35\linewidth}
			\centering
			$\begin{bNiceArray}{cc:cc}[first-row,last-row=5,first-col,last-col,nullify-dots]
				& 0 & 1 & 2 & 3 & \\
				0 & a & b & a & c & \\
				1 & d & c & d & c & \\
				\hdottedline
				2 & a & b & a & b & \\
				3 & a & c & d & b & \\
				&  &  & &  &
			\end{bNiceArray}$
		\end{minipage}%
	}%
	\subfigure[The matrix $U^{\order}$ with the order {$\order: [x_{1}, x_{0}]$} ]{
		\begin{minipage}[t]{0.325\linewidth}
			\centering
			$\begin{bNiceArray}{cc:cc}[first-row,last-row=5,first-col,last-col,nullify-dots]
				& 0 & 1 & 2 & 3 & \\
				0 & a & a & b & c & \\
				1 & a & a & b & b & \\
				\hdottedline
				2 & d & d & c & c & \\
				3 & a & d & c & b & \\
				&  &  & &  &
			\end{bNiceArray}$
		\end{minipage}%
	}%
	\subfigure[The matrix $ U^{\order'} $ with the linear transformation {$\order' : [x_{0}\oplus x_{1}, x_{1}]$} ]{
		\begin{minipage}[t]{0.3\linewidth}
			\centering
			$\begin{bNiceArray}{cc:cc}[first-row,last-row=5,first-col,last-col,nullify-dots]
				& 0 & 1 & 2 & 3 & \\
				0 & a & c & a & b & \\
				1 & a & b & d & c & \\
				\hdottedline
				2 & a & b & a & b & \\
				3 & d & c & d & c & \\
				&  &  & &  &
			\end{bNiceArray}$
		\end{minipage}%
	}%
\end{figure*}
\begin{figure*}[htbp]
	\subfigure[The QMDD for the matrix $U$]{
		\begin{minipage}[t]{0.35\linewidth}
			\centering
			\includegraphics[scale=0.4]{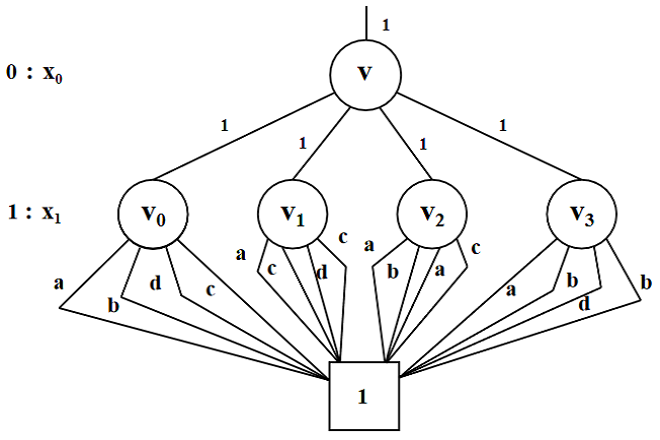}
		\end{minipage}%
	}%
	\subfigure[The QMDD for the matrix $U^{\order}$ ]{
		\begin{minipage}[t]{0.325\linewidth}
			\centering
			\includegraphics[scale=0.4]{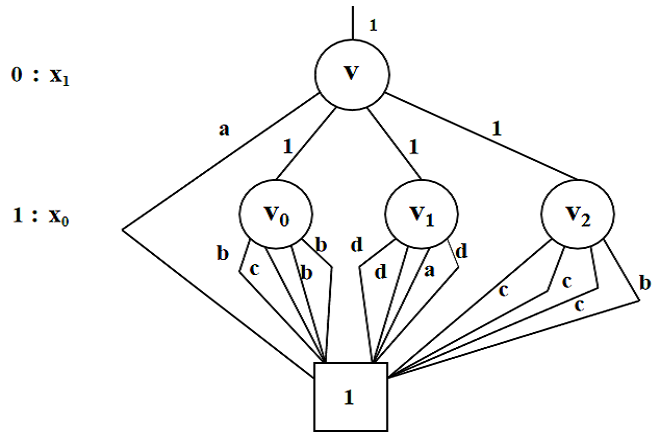}
		\end{minipage}%
	}%
	\subfigure[The LTQMDD for the matrix $U^{\order'}$]{
		\begin{minipage}[t]{0.3\linewidth}
			\centering
			\includegraphics[scale=0.4]{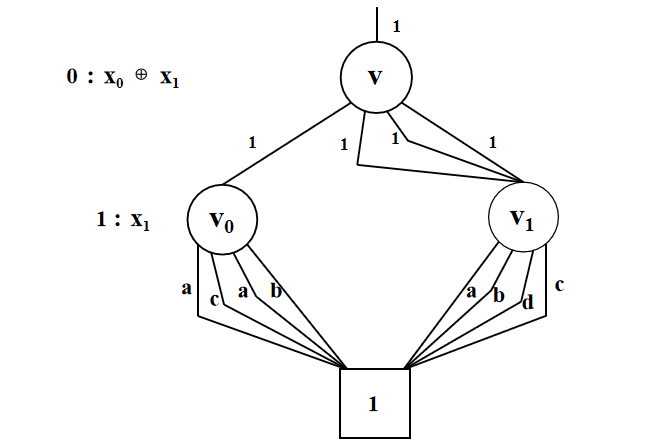}
		\end{minipage}%
	}%
	\caption{The matrices and QMDDs representing the same quantum operation with different linear transformation order}
	\label{fig:MatricesQMDD2}
\end{figure*}

\looseness=-1
Clearly, the size of unitary matrix is exponential in the number of qubits.
In practice, matrices are spare and contain many identical submatrices.
By making advantage of these characteristics of matrices in practical applications, representing unitary matrices in a compact form become feasible.
Recently, \citet{NieWM2016} proposed a compact and canonical representation of matrices, namely, quantum multiple-valued decision diagrams (QMDDs).
The basic idea of QMDDs is to iteratively decompose the matrix into four submatrices according to a variable.

\begin{definition}
	A \textit{quantum multiple-valued decision diagram \\
	(QMDD)} is a rooted directed acyclic graph with a root edge $(v_r, e_r)$ where $v_r$ is the root node and $e_r$ is an edge pointing to $v_r$.
	The nodes are classified into two types: \textit{internal} and \textit{terminal}.
	Each internal node $v$ is associated with an index $\ind[v]$ where $\ind[v] \in \set{0, \cdots, n-1}$.
	The node $v$ has four edges $e_{0}(v), e_{1}(v), e_{2}(v)$ and $e_{3}(v)$ pointing to the successors $v_{0}, v_{1}, v_{2}$ and $v_{3}$, respectively.
	Each edge $e$, including the root edge $e_r$, is labeled with a complex-valued weight $w(e)$.	
	The only terminal node is labeled with $1$ and has no successor.
	The index of terminal node is $n$. \qed
\end{definition}

In a QMDD, a node is at the $i$-th level, if its index is $i$.
A complete path in a QMDD is a path from the root edge to the terminal node. 
Throughout this paper, we require that the indices of internal nodes on all complete paths in the QMDD appear in an increasing order.
The size of a QMDD $G$, written $\size[G]$, is the number of its internal nodes.

We further give the semantics of QMDDs, that is, a function mapping QMDDs together with indices to matrices.
\begin{definition}\label{def:semantics}
	Let $(v, e)$ be a QMDD and $j$ an index where $j \leq \ind[v]$. 
	The matrix $M_{(v, e)}^j$ represented by $(v, e)$ with $j$ is defined as: 
	\begin{enumerate}
		\item If $v$ is the terminal node, then $M_{(v, e)}^j = [w(e)]_{s \times s}$		
		where $[w(e)]_{s \times s}$ denotes the matrix of size $s \times s$ where each entry is $w(e)$ and $s = \ind[v]- j$.
		
		\item If $v$ is an internal node, then		
		\begin{equation*}
			M_{(v, e)}^j = [w(e)]_{s \times s} \otimes \begin{bmatrix}
				M^{k}_{(v_0, e_0)} & M^{k}_{(v_1, e_1)} \\
				M^{k}_{(v_2, e_2)} & M^{k}_{(v_3, e_3)}
			\end{bmatrix}
		\end{equation*} 
		
		where $[w(e)]_{s \times s}$ denotes the matrix of size $s \times s$ where each entry is $w(e)$ and $s = \ind[v]- j$, $k = \ind[v] - 1$ and $M_{(v_i, e_i)}^k$ is the matrix represented by the QMDD $(v_i, e_i)$ with $k$. \qed
		
	\end{enumerate} 
	
\end{definition}

%


\begin{example}
	Two natural numbers 0 and 1 denote the two Boolean vectors $ (00)_{2} $ and $ (01)_{2} $ respectively in which the Boolean value of $ x_{0} $ is $ 0 $. Similarly, the Boolean value of $ x_{0} $ is $ 1 $ in the two natural numbers 2 and 3. 
	According to $x_{0}$, decomposing the matrix $ U $ in Figure \ref{fig:MatricesQMDD2}(a) leads to four submatrices $U_0 = \begin{bmatrix}
		a & b \\
		d & c
	\end{bmatrix}$, $U_1 = \begin{bmatrix}
		a & c \\
		d & c
	\end{bmatrix} $, $U_2 = \begin{bmatrix}
		a & b \\
		a & c 
	\end{bmatrix}$, and $U_3 = \begin{bmatrix}
		a & b \\
		d & b 
	\end{bmatrix}$.
	The submatrix $U_1$ contains all entries where the value of $x_0$ is $0$ in the row index $i$ and that of $x_0$ is $1$ in the column index $j$.
	The other submatrices $U_0$, $U_2$ and $U_3$ are similar.
	Finally, $U_1$ is recursively decomposed into four submatrices of size $1 \times 1$: $[a]$, $[c]$, $[d]$ and $[c]$.
	
	Figure \ref{fig:MatricesQMDD2}(d) shows the QMDD representing the matrix $U$.
	The weights of the root edge and four outgoing edges of $v$ are $1$.	
	According to the decomposition rule, we obtain the subgraph $(v_0, e_0)$ representing the submatrix $U_0$.
	The node $v_0$ has four outgoing edges $a$, $b$, $d$ and $c$ pointing to the terminal node, which denotes the entries each of whose value is the weight of the corresponding edge. \qed
\end{example}


To compactly represent unitary matrices, a QMDD should be compressed to another equivalent one according to a set of reduction rules.
We say two internal nodes $v$ and $v'$ are isomorphic if they are labeled with the same variable and have all corresponding edges pointing to the same nodes with the same weight.

\begin{itemize}
	\item[\textbf{RI}] Eliminate an internal node $v$ isomorphic to a distinct node $v'$, and redirect every incoming edge of $v$ to $v'$.
	\item[\textbf{RS}] Eliminate a node $v$ such that all outgoing edges point to the same node $v'$ and have the same weight $w'$, and redirect every incoming edge $e$ of $v$ to $v'$ and set the weight of $e$ to be $w \cdot w'$ where $w$ is the original weight of $e$.
\end{itemize}

A QMDD is \textit{reduced}, if none of the rules \textbf{RI} and \textbf{RS} can be applied in it.

Canonicity is a desired feature of representations of unitary matrices, \ie, any unitary matrix has a unique representation.
This feature is of particular importance.
On the one hand, a more compact representation can be obtained since there does not exist two distinct structures denoting the same submatrix.
On the other hand, equivalence checking, the commonly used query task, can be accomplished under canonical representations in a constant time.
To obtain the canonicity feature for QMDDs, it is necessary to normalize the weights of all outgoing edges of internal nodes.
A QMDD is \textit{normalized}, if for every internal node, the largest magnitude of all non-zero weights of the outgoing edges is $1$.
We remark that if two or more edges have weights of the largest magnitude, then we require only the weight of the leftmost edge to be $1$.
By imposing ordering, reduction and normalization properties, QMDDs become a canonical form of unitary matrices \cite{NieWM2016,ZulHR2019}.

\subsection{Linear Transformation}
\looseness=-1
For a subset $\varSetY$ of variables, its \textit{linear combination} is $\bigoplus_{x \in \varSetY} x$.
We say a set $\Pi$ of linear combinations is \textit{fully ranked}, if every variable $x \in \varSet$ can be represented by an exclusive-or of a subset of $\Pi$.
A \textit{linear transformation} is a mapping $\pi$ from $\set{0, \cdots, n-1}$ to exactly one element of a fully ranked set $\Pi$ of linear combinations.
We use $\order_i$ for the $i$-th element of $\order$.
The \textit{standard order} denotes the increasing order of variables $[x_0, x_1, \cdots, x_{n-1}]$. 
A linear transformation can be used to represent an automorphism $\order: \bool^n \rightarrow \bool^n$ where the $i$-th Boolean value of $\order(v)$ is the value of the $i$-th linear combination $\order_i$ under $v$ (\ie, $\order(v)_i = \order_i(v)$ for every $v \in \bool^n$).

\begin{example} \label{exm:linearTransOrder}
	Suppose that $\varSetX = \set{x_0, x_1, x_2}$.
	The variable order $\order^1 = [x_2, x_0, x_1]$ is a linear transformation since each variable is an element of the order.
	Because $x_0 \equiv (x_0 \oplus x_1) \oplus (x_1 \oplus x_2) \oplus x_2$ and $x_1 \equiv (x_1 \oplus x_2) \oplus x_2$, so $\order^2 = [x_0 \oplus x_1, x_1 \oplus x_2, x_2]$ is also a linear transformation.
	However, the order $\order^3 = [x_0 \oplus x_1, x_1 \oplus x_2, x_0 \oplus x_2]$ is not as any $x_i$ can not be represented by the linear combination of a subset of $x_0 \oplus x_1$, $x_1 \oplus x_2$, $x_0 \oplus x_2$.
	
	For a Boolean vector $v = (001)_2 = 1$, it means that $x_0 = 0$, $x_1 = 0$ and $x_2 = 1$ under the standard order $[x_0, x_1, x_2]$.
	It follows that $x_0 \oplus x_1 = 0$ and $x_1 \oplus x_2 = 1$.
	Hence, $\order^1$ maps $v$ to the different Boolean vector $\order^1(v) = (100)_2 = 4$.
	Similarly, $\order^2(v) = (011)_2 = 3$. \qed
\end{example}

\section{QMDDs with Linear Transformation}
In a traditional way, the unitary matrix of a quantum operation is constructed following the increasing variable order $[x_0, x_1, \cdots, \\ x_{n - 1}]$.
It is possible to obtain a more compact QMDD via different variable orders.
We illustrate this with the following example.

\begin{example} \label{exm:QubitsAndMatrix2}
	As Figure \ref{fig:MatricesQMDD2}(a) shows, $U$ is a unitary matrix which directly represents a quantum operation since it is based on the standard order $[x_0, x_1]$.
	The QMDD for $U$ is shown in Figure \ref{fig:MatricesQMDD2}(d).
	Since the four submatrices of $U$ based on $x_0$ is distinct, the QMDD for $U$ has size $5$.
	We now consider the variable order $\order: [x_1, x_0]$.
	Under the variable order $\order$, the natural number $1$ corresponds to the Boolean vector $(01)_2$, and indicates that the values of $x_1$ and of $x_0$ is $0$ and $1$, respectively.
	The natural number $1$ under $\order$ in fact corresponds to $2$ under the standard order $[x_0, x_1]$.
	Similarly, $0, 2$ and $3$ under $\order$ corresponds to $0, 1$ and $3$ under the standard order, respectively.	
	Under the new variable order, it is necessary to adjust the position of some entries of $U$.
	By first exchanging the $1$-st and $2$-nd rows of $U$, and then exchanging the $1$-st and $2$-nd columns, we obtain the matrix $U^{\order}$ based on $[x_1, x_0]$.
	Its submatrix $U^{\order}_0$ is $[a]_{2 \times 2}$, which can be represented by the terminal node and a edge with weight $a$.
	The QMDD for $U^{\order}$, shown in Figure \ref{fig:MatricesQMDD2}(e), has size $4$ smaller than $U$. \qed
\end{example}

Variable orders can be considered as an automorphism on $\bool^n$. 
According to an automorphism, the original matrix can be converted into another one by rearranging some entries' positions.
The new matrix represents the same quantum operation but may contain more identity submatrices, resulting in a smaller QMDD than before.


%

Linear transformations are a class of automorphisms on $\bool^n$ that express much more automorphisms than variable order but enjoy efficient representation \cite{MeiST2000}.
Inspired by the concept of linear transformations, in this section, we will introduce the integration of matrices with linear transformations, then propose a compact and canonical representation of quantum operations, namely Linearly Transformed Quantum Multipled-Decision Diagram (LTQMDD), and finally devise a minimization algorithm for LTQMDDs, which essentially finds a locally optimal linear transformation for LTQMDDs. 

\subsection{Incorporating Linear Transformation}
Traditionally, quantum states are represented by a normalized vector of size $2^{n}$ and quantum operations are represented by a unitary matrix of size $2^n \times 2^n$. 
In the following, we incorporate the vector-based representation of quantum states and the matrix-based representation of quantum operations with an additional linear transformation.
The basic idea is to rearrange the position of every entry in the vector and the matrix according to the linear transformation.

\begin{definition} \label{def:LTVector}
	Let $\order$ be a linear transformation over $\varSet$ and $v$ a vector of size $2^n$.
	The \textit{linearly transformed vector} of $v$ by $\order$, written $v^{\order}$, is defined as $v^{\order}_{i} = v_{\order(i)}$ for $0 \leq i < 2^n$. \qed
\end{definition}

\begin{definition} \label{def:LTMatrix}
	Let $\order$ be a linear transformation over $\varSet$ and $U$ a matrix of size $2^n \times 2^n$.
	The \textit{linearly transformed matrix} of $U$ by $\order$, written $U^{\order}$, is defined as 
	\[
	u^{\order}_{i, j} = u_{\order(i), \order(j)} \text{ for } 0 \leq i, j < 2^n
	\]
	where $u^{\order}_{i, j}$ is the entry of $U^{\order}$ with the $i$-st row and $j$-st column and $u_{\order(i), \order(j)}$ is the entry of $U$ with the $\order(i)$-st row and $\order(j)$-st column. \qed
\end{definition}


%
%

We remark that the vector and matrix are identical to the traditional ones when the linear transformation is the standard order.

\begin{example}
	Let $\order$ be a linear transformation $[x_0 \oplus x_1, x_1]$.
	It is easily verified that $\order(0) = 0$, $\order(1) = 3$, $\order(2) = 2$ and $\order(3) = 1$.
	The quantum state $\quanState[\varSet]$ is a vector $v  = [\alpha_{0} \ \alpha_{1} \ \alpha_{2} \ \alpha_{3}]^{\top}$.
	The linearly transformed vector of $v$ by $\order$ is $[\alpha_{0} \ \alpha_{3} \ \alpha_{2} \ \alpha_{1}]^{\top}$.
	
	Figure \ref{fig:MatricesQMDD2}(a) shows a quantum operation represented by the unitary matrix $U$ and Figure \ref{fig:MatricesQMDD2}(c) shows the linearly transformed matrix $U^{\order'} = [u^{\order'}_{i, j}]_{4 \times 4}$ of $U$ by $\order'$,  which denotes the same quantum operation as $U$.
	Since $\order'(1) = 3$ and $\order'(3) = 1$, the elements of $U^{\order'}$ with $1$-st row are as follows: $u^{\order'}_{10} = u_{30} = a$, $u^{\order'}_{11} = u_{33} = b$, $u^{\order'}_{12} = u_{32} = d$ and $u^{\order'}_{13} = u_{31} = c$. \qed
\end{example}

\subsection{Linearly Transformed QMDDs}
\looseness=-1
We say a QMDD $(v, e)$ together with a linear transformation $\order$ is called \textit{linearly transformed QMDD} (LTQMDD) $(v, e, \order)$.
A LTQMDD $(v, e, \order)$ respects the linear transformation $\order$.
In the following, we will define the semantics of LTQMDDs (\ie, the mapping from LTQMDDs to quantum operations) that serves as the theoretical foundation of LTQMDD-representations for quantum computing.
Remind that, we have given the semantics for QMDDs (cf. Definition \ref{def:semantics})
It is easy to extend the semantics for QMDDs to LTQMDDs via the concept of linearly transformed matrices (cf. Definition \ref{def:LTMatrix}).

\begin{definition} \label{def:semLTQMDD}
	Let $\order$ be a linear transformation over $\varSet$, $(v, e, \order)$ a LTQMDD and $j$ the index s.t. $j \leq \ind[v]$.
	The semantics for the LTQMDD $(v, e, \order)$ with $j$ is the linearly transformed matrix $(M^j_{(v, e)})^{\order}$ where $M^j_{(v, e)}$ denotes the matrix represented by $(v, e)$ with $j$. \qed 
\end{definition}



\begin{example} \label{exm:LTQMDD}
	Under the linear transformation $\order': [x_0 \oplus x_1, x_1]$, the linearly transformed matrix $U^{\order'}$ is illustrated in Figure \ref{fig:MatricesQMDD2}(c).
	The submatrix $U^{\order'}_0$ is $\begin{bmatrix}
		a & c \\
		a & b
	\end{bmatrix}$ and $U^{\order'}_1$ is $\begin{bmatrix}
		a & b \\
		d & c
	\end{bmatrix}$.
	Two submatrices $U^{\order'}_2$ and $U^{\order'}_3$ are identical to $U^{\order'}_1$.	
	As Figure \ref{fig:MatricesQMDD2}(f) shows, the corresponding LTQMDD for $U^{\order'}$ contains a root node $v$ with four outgoing edges to two successor nodes $v_0$ and $v_1$. 
	The first outgoing edge $e_0(v)$ together with $v_0$ denotes the submatrix $U^{\order'}_0$.
	The other three outgoing edges $e_1(v)$, $e_2(v)$ and $e_3(v)$ together with $v_1$ denote the submatrices $U^{\order'}_1$, $U^{\order'}_2$ and $U^{\order'}_3$, respectively.
	It is easily observed that the size of the LTQMDD for the linearly transformed matrix $U^{\order'}$ is smaller than the two QMDDs denoting the same quantum operation shown in Figure \ref{fig:MatricesQMDD2}(d) and \ref{fig:MatricesQMDD2}(e). \qed
\end{example}

The number of variable orders is $n!$.
In contrast, the number of linear transformations is $\prod_{i = 0}^{n - 1} (2^n - 2^i)$ \cite{MeiT2001}.
Linear transformation is able to convert a quantum matrix into one with more identical submatrices compared to variable orders, and hence enable us to gain a more compact representation of quantum functionality.
This was verified by Example \ref{exm:LTQMDD}.
It is easily verified that extending QMDDs with linear transformation does not affect the canonicity property.
\begin{theorem}
	For a linear transformation $\order$, any unitary matrix $U$ has a unique reduced and normalized LTQMDD respecting $\order$.
\end{theorem}

	\begin{proof}
		By Definition \ref{def:semLTQMDD}, the QMDD-representation of $U^\order$ is identical to the LTQMDD-representation of $U$.
		By Theorem 2 in \cite{NieWM2016}, the matrix $U^\order$ has a unique reduced and normalized QMDD.
		So the matrix $U$ has a unique reduced and normalized LTQMDD respecting $\order$.
%
	\end{proof}

In addition, a quantum state can be transformed via not only one quantum operation but also a sequence of quantum operations. 
In order to manipulate the combination of quantum operations, we will use three operations: addition ($+$), multiplication ($\cdot$) and Kronecker product ($\otimes$) of two unitary matrices.
\citet{NieWM2016} designed three corresponding algorithms for the above operations over QMDDs.
These algorithms can be directly applied in LTQMDDs without any modification.
We do not present these algorithms, for details, please refer to \cite{NieWM2016}.

%

\begin{figure*}[t]
	\subfigure[The original matrix $U$]{
		\begin{minipage}[t]{0.225\linewidth}
			\centering
			$\begin{bNiceArray}{cc:cc}[first-row,last-row=5,first-col,last-col,nullify-dots]
				& 0 & 1 & 2 & 3 & \\
				0 & M_{00} & M_{01} & M_{02} & M_{03} & \\
				1 & M_{10} & M_{11} & M_{12} & M_{13} & \\
				\hdottedline
				2 & M_{20} & M_{21} & M_{22} & M_{23} & \\
				3 & M_{30} & M_{31} & M_{32} & M_{33} & \\
				&  &  & &  &
			\end{bNiceArray}$
		\end{minipage}%
	}%
	\hspace*{1mm}
	\subfigure[The matrix $U'$ with the order {$\order'\hspace*{-2mm}: \hspace*{3mm} [\order_0, \cdots, \order_{i + 1}, \order_i, \cdots, \order_{n - 1}]$}]{
		\begin{minipage}[t]{0.225\linewidth}
			\centering
			$\begin{bNiceArray}{cc:cc}[first-row,last-row=5,first-col,last-col,nullify-dots]
				& 0 & 1 & 2 & 3 & \\
				0 & M_{00} & M_{02} & M_{01} & M_{03} & \\
				1 & M_{20} & M_{22} & M_{21} & M_{23} & \\
				\hdottedline
				2 & M_{10} & M_{12} & M_{11} & M_{13} & \\
				3 & M_{30} & M_{32} & M_{31} & M_{33} & \\
				&  &  & &  &
			\end{bNiceArray}$
		\end{minipage}%
	}%
	\hspace*{1mm}
	\subfigure[The matrix $U''$ with the order {$\order''\hspace*{-2mm}: \hspace*{3mm} [\order_0, \cdots, \order_i \oplus \order_{i + 1}, \order_i, \cdots, \order_{n - 1}]$}]{
		\begin{minipage}[t]{0.225\linewidth}
			\centering
			$\begin{bNiceArray}{cc:cc}[first-row,last-row=5,first-col,last-col,nullify-dots]
				& 0 & 1 & 2 & 3 & \\
				0 & M_{00} & M_{02} & M_{03} & M_{01}& \\
				1 & M_{20} & M_{22} & M_{23} & M_{21} & \\
				\hdottedline
				2 & M_{30} & M_{32} & M_{33} & M_{31}& \\
				3 & M_{10} & M_{12} & M_{13} & M_{11} & \\
				&  &  & &  &
			\end{bNiceArray}$
		\end{minipage}%
	}%
	\hspace*{1mm}
	\subfigure[The matrix $U^*$ with the order {$\order^*\hspace*{-2mm}: \hspace*{3mm} [\order_0, \cdots, \order_{i + 1}, \order_i \oplus \order_{i + 1}, \cdots, \order_{n - 1}]$}]{
		\begin{minipage}[t]{0.225\linewidth}
			\centering
			$\begin{bNiceArray}{cc:cc}[first-row,last-row=5,first-col,last-col,nullify-dots]
				& 0 & 1 & 2 & 3 & \\
				0 & M_{00} & M_{03} & M_{01} & M_{02}& \\
				1 & M_{30} & M_{33} & M_{31} & M_{32} & \\
				\hdottedline
				2 & M_{10} & M_{13} & M_{11} & M_{12}& \\
				3 & M_{20} & M_{23} & M_{21} & M_{22} & \\
				&  &  & &  &
			\end{bNiceArray}$
		\end{minipage}%
	}%
\end{figure*}
\begin{figure*}[htbp]
	\subfigure[The QMDD for the matrix $U$]{
		\begin{minipage}[t]{0.5\linewidth}
			\centering
			\includegraphics[scale=0.4]{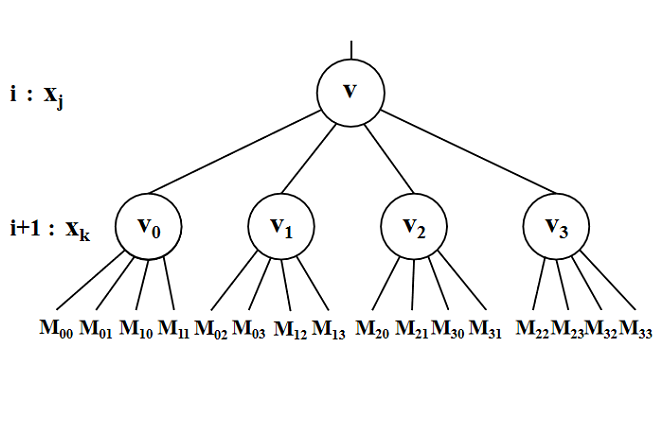}
		\end{minipage}%
	}%
	\subfigure[The LTQMDD for the matrix $U'$ ]{
		\begin{minipage}[t]{0.5\linewidth}
			\centering
			\includegraphics[scale=0.4]{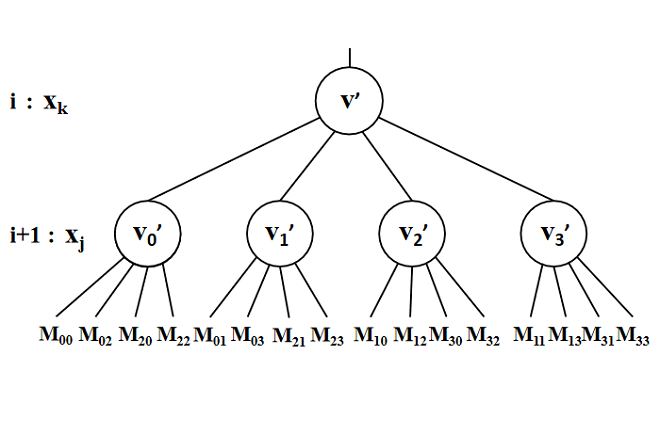}
		\end{minipage}%
	}%
	
	\subfigure[The LTQMDD for the matrix $U''$]{
		\begin{minipage}[t]{0.5\linewidth}
			\centering
			\includegraphics[scale=0.4]{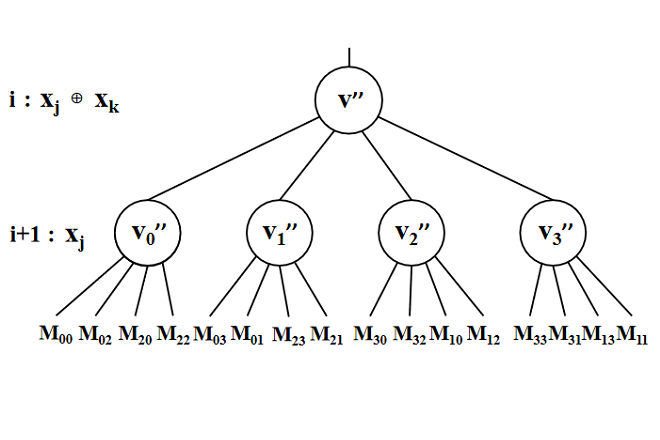}
		\end{minipage}%
	}%
	\subfigure[The LTQMDD for the matrix $U^*$ ]{
		\begin{minipage}[t]{0.5\linewidth}
			\centering
			\includegraphics[scale=0.4]{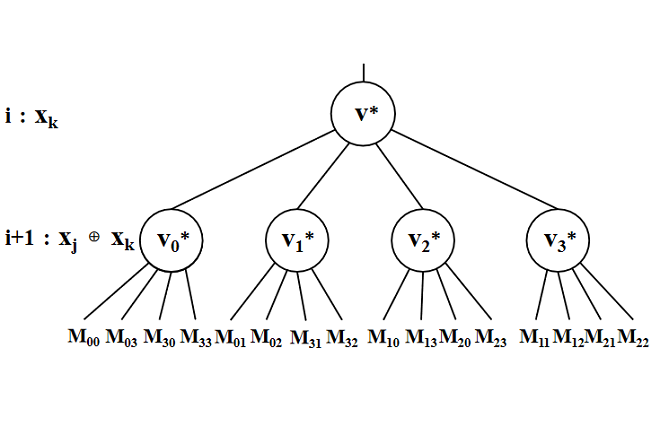}
		\end{minipage}%
	}%
	\caption{The three level exchange procedures}
	\label{fig:MatricesQMDD3}
\end{figure*}

\subsection{Linear Sifting of QMDDs}
Choosing an appropriate linear transformation plays a crucial role in reducing the size of QMDDs.
To identify a good linear transformation for a unitary matrix, we propose linear sifting algorithm for QMDDs.

Firstly, we introduce three level exchange procedures that are essential components of linear sifting algorithm.
Assume that $\order$ is a linear transformation with the $i$-th element $\order_i$ and the $(i + 1)$-th element $\order_{i + 1}$.
The standard level exchange procedure for levels $i$ and $i + 1$, proposed by \citet{NieWM2016}, obtains a QMDD with the order $\order': [\order_0, \cdots, \order_{i + 1}, \order_i, \cdots, \order_{n - 1}]$.
For each node $v$ with index $i$, we use $v'$ to denote the node corresponding to $v$ in the new LTQMDD.
For each $0 \leq l, m \leq 3$, we let the edge $e_l(v'_m)$ point to the node $v_{\order'(m)}$ and let the weight $w(e_l(v'_m))$ to be the weight $w(e_{\order'(l)}(v_{\order'(m)}))$.
Secondly, the upper and lower level exchange procedures gain LTQMDDs with the order $\order'': [\order_0, \cdots, \order_i \oplus \order_{i + 1}, \order_i, \cdots, \order_{n - 1}]$ and $\order^{*}: [\order_0, \cdots, \order_{i + 1}, \order_i \oplus \order_{i + 1}, \cdots, \order_{n - 1}]$, respectively.
The processes of them are the same as the standard level exchange procedure except that we use $\order''$ and $\order^{*}$ instead of $\order'$.

\looseness=-1
Figure \ref{fig:MatricesQMDD3} depicts the three level exchange procedures in terms of matrices.
Due to the space limit, we only illustrate the upper level exchange procedure here.
Taking the order $\order'': [\order_0, \cdots, \order_i \oplus \order_{i + 1}, \order_i, \cdots, \order_{n - 1}]$ as an example.
Let $\order_i = x_j$ and $\order_{i + 1} = x_k$.
It follows that $\order_i \oplus \order_{i + 1} = x_j \oplus x_k$.
As shown in Figure \ref{fig:MatricesQMDD3}(a), the original matrix $U$ is decomposed into $16$ submatrices w.r.t. $[x_j, x_k]$.
According to $[x_j \oplus x_k, x_k]$, the matrix becomes $U''$ shown in Figure \ref{fig:MatricesQMDD3}(c).
Figure \ref{fig:MatricesQMDD3}(g) depicts the LTQMDD with the root node $v''$ representing $U''$.
The matrix $U''$ is obtained from the original matrix $U$ by first exchanging $1$-st row and $2$-nd row, then exchanging $1$-st row and $3$-rd row, and finally exchange the corresponding columns as we do on rows.
For example, $M_{33}$ is the entry of $U''$ with the $1$-st row and the $1$-st column.
The $l$-th outgoing edge from the $m$-th successor node of $v''$ points to an internal node that represents the submatrix of $U''$ with the $m$-th row and the $l$-th column.
Notably, the upper level exchange procedure can be reversed by the lower one, and vice versa. 
The inverse operation of standard exchange procedure is itself.

\looseness=-1
The new LTQMDD generated by the level exchange procedures may violate the reducing or normalization property.
It is necessary to obtain a smaller LTQMDD via the reduction and normalization rules after each level exchange procedure.
The efficiency of level exchange procedures is due to its local operations on the nodes of the two involved levels.
However, the normalization rule requires adjusting the whole LTQMDD in the worst case.
To fix this defect, we adopt the approach proposed by \citet{NieWM2016}.
The basic idea is to save the change of weights in the nodes instead of propagating them to incoming edges.
For details, please refer to \cite{NieWM2016}.

\looseness=-1
Armed with the three level exchange procedures, we are ready to introduce linear sifting algorithm.
It finds the locally optimal position of each level but also the locally optimal linear combination.
It sorts levels into a decreasing sequence according to the number of nodes at each level.
For each level $i$ following the decreasing sequence, the locally optimal position and linear combination of the chosen level $i$ will be determined in the following way.
The movement of the chosen level consists of two phases.
The direction of movement of the two phases is determined based on its initial position. 
Here, we only consider the case where level $i$ is close to the bottom level \ie, $i > \frac{n}{2}$.
The process of the case where level $i$ is close to the top level is similar.
In the 1st phase, level $i$ is first moved to the bottom level and then returns to the initial position with the initial linear combination.
In the 2nd phase, level $i$ will be moved towards to the top level and finally achieves the locally optimal position with the corresponding linear combination.
To recover the position and linear combination, we maintain four sequences of level exchange procedures $\tau^1_{init}$, $\tau^1_{opt}$, $\tau^2_{init}$ and $\tau^2_{opt}$.
We also use $s^1_{opt}$ and $s^2_{opt}$ for the minimal size of LTQMDDs discovered by the 1st and 2nd phases, respectively.
Each element of the sequences contains the type of level exchange and the index of level which be interchanged by level $i$.
In each element, we use $s$, $u$ and $l$ for the standard, upper and lower level exchange, respectively.
The sequence $\tau^1_{init}$ records the sequence of level exchange procedures from the initial position to the bottom for the 1st phase while $\tau^1_{opt}$ is used to achieve the optimal position discovered by the 1st phase.
The sequences $\tau^2_{init}$ and $\tau^2_{opt}$ are similar but for the 2nd phase.

\looseness=-1
In the 1st phase, level $i$ is moved to the bottom level by level via the standard or upper level exchange procedure and record two sequences of level exchange procedures $\tau^1_{init}$ and $\tau^1_{opt}$.
Suppose the current linear transformation is $\order: [\cdots, \order_i, \order_{j}, \cdots]$.
The steps of each move work as follows:
\begin{enumerate}
	\item 
	Assume that the standard level exchange procedure is used for the selected $i$ and $j$.
	The corresponding linear transformation becomes $\order': [\cdots, \order_{j}, \order_i, \cdots]$.
	Let $s'$ be the size of the new LTQMDD.
	
	\item 
	Assume that the upper level exchange procedure is used for the selected $i$ and $j$.
	The corresponding linear transformation becomes $\order'': [\cdots, \order_i \oplus \order_{j}, \order_i, \cdots]$.
	Let $s''$ be the size of the new LTQMDD.
	
	\item 
	If $s' > s''$, then we chose the LTQMDD with the order $\order''$ and add $(u, j)$ to $\tau^1_{init}$; otherwise, we chose the LTQMDD with the order $\order'$ and add $(s, j)$ to $\tau^1_{init}$.
	
	\item If $\min\set{s', s''} < s^1_{opt}$, then $s^1_{opt} = \min\set{s', s''}$ and let $\tau^1_{opt}$ be $\tau^1_{init}$.
\end{enumerate}
During this moving process, the linear combination of variables may be introduced since the $i$-th element of $\order$ will be $\order_i \oplus \order_j$ when the upper level exchange procedure generates a better LTQMDD.
Let $\tau'_{init}$ be the reverse sequence of $\tau^1_{init}$.
Level $i$ returns to its initial position with the initial linear combination via performing the inverse operation of each level exchange procedure of $\tau'_{init}$ one by one.

\looseness=-1
Now it turns to the 2nd phase.
This phase is similar to the 1st one which moves level $i$ to the bottom except for the following:
Firstly, we choose its predecessor level $k$ rather than its successor level $j$ when executing the level exchange procedure in each move.
Secondly, we record the two sequences of level exchange procedures as $\tau^2_{init}$ and $\tau^2_{opt}$, and the minimal size of LTMQDDs as $s^2_{opt}$.
Finally, we move level $i$ to its locally optimal position with the initial linear combination.
It firstly moves to the initial position with the initial linear combination.
If $s^1_{opt} < s^2_{opt}$, then the locally optimal position and linear combination are found in 1st phase.
In this case, we execute each level exchange procedure of $\tau^1_{opt}$ one by one.
Otherwise, the sequence we perform is $\tau^2_{opt}$.

\section{Experimental Results}
\looseness=-1
We have implemented the linear sifting based on the publicly-available JKQ-framework \cite{WilR2020} which includes the state-of-the-art QMDD package \cite{ZulHR2019} and compilation method \cite{BurL2021}. 
We use a Boolean matrix with size $n \times n$ to store the linear transformation $\order$.
The variable $x_j$ is in the linear combination $\order_i$ iff the entry of the Boolean matrix with the $i$-th row and $j$-th column is $1$.

\looseness=-1
We use benchmark circuits from Qiskit \cite{Anis2021}, QASMBench \cite{LiA2021}, Feynman \cite{Amy2019} and GRCS (Google Random Circuit Sampling Benchmarks) \cite{BoiIS2018}.
Since GRCS is too large to complete the compilation process, we choose parts of the circuit as test cases. 
The name "GRCS\_i\_j\_k" denotes the $k$-th circuit of $i$ qubits with the first $j$ operations.
We firstly compile each test case into a QMDD with the standard variable order.
Then, we apply the (linear) sifting algorithm in the compiled QMDD until it converges.
Furthermore, we consider only benchmarks from these comparisons if (1) the QMDD can be compiled using the standard order within $1$ hour time limit and $10$GB memory limit, and (2) the resultant QMDD generated by the standard variable order has the size of more than $150$.
Finally, there are $63$ test cases that meet the above conditions, with results presented in Table \ref{tab:exp}.
The machine running the benchmark is equipped with Intel Xeon Gold 6248R 3.00GHz CPU and 128GB memory.

\looseness=-1
The comparison between original sifting and linear sifting algorithms are shown in Table \ref{tab:exp}.
The columns "Qubits" and "Gates" represent the number of qubits and the number of gates, respectively.
The column "Standard" denotes the size of complied QMDDs with the standard order.
The columns "Sifting" and "Linear sifting" denote the results of the corresponding algorithms respectively and contain two subcolumns where "Size" denotes the number of nodes of the compiled QMDDs (LTQMDDs) and "Time" is the total runtime in seconds.
The column "Ratio" indicates the size improvement of the linear sifting compiled LTQMDDs over the sifting QMDDs.

We can make several observations from Table \ref{tab:exp}.
Firstly, both sifting algorithms dramatically reduce the size of the initial QMDDs by an average of more than $80\%$.
In addition, linear sifting generates the LTQMDDs with total number of nodes $11\%$ smaller than the QMDDs that provided by sifting. 
Among the $63$ test cases mentioned above, $36$ test cases using linear sifting achieved smaller sizes and only $15$ test cases have larger sizes than sifting.
For the test cases: "csum\_mux\_9", "ham15-low", "ham15-med", "vqe\_uccsd\_n6" and "grcs\_20\_100\_2", the size of LTQMDDs generated from linear sifting are $70\%$, $42\%$, $57\%$, $50\%$ and $41\%$ smaller than that of sifting.
%
In particular, for "csum\_mux\_9", sifting provide the same QMDD as the initial one, but linear sifting obtains $70\%$ improvement on sizes over the initial QMDD.
Apart from the perspective of sizes, we can see that linear sifting algorithm is slower than original sifting algorithm in most instances.
The reason is as follow.
In each move, linear sifting calls at least one more level exchange procedure compared to original sifting.
If two sifting algorithms produce QMDDs of almost equal sizes, then linear sifting takes twice as long as sifting.
However, in most instances, linear sifting produces a LTQMDD smaller than sifting.
Therefore, linear sifting takes only $1.7$ times longer than original sifting on total time. 
In addition, the time complexity of the operations on QMDDs depends on its size.
Subsequent operations can benefit from the compact representation.
Hence, it is worthy of costing more time to generate more compact QMDDs.
\looseness=-1

\begin{center}
	\begin{table*}
		\small
		\centering
		\caption{Experimental results for sifting and linear sifting}
		\label{tab:exp}
		\begin{tabular}{|c|c|c|c|cc|cc|c|}
			\hline
			\vspace*{-0.6mm}
			\multirow{2}{*}{Circuit}          & \multirow{2}{*}{Qubits} & \multirow{2}{*}{Gates} & Standard  & \multicolumn{2}{c}{Sifting} \vline & \multicolumn{2}{c}{Linear Sifting} \vline & \multirow{2}{*}{Ratio} \\
											  &                         &                        & Size      & Size            & Time      & Size                  & Time       &                        \\ \hline
			bigadder\_n18                & 18 & 5    & 3128  & 78            & 7.77  & 78             & 4.56  & 1             \\
			csla\_mux\_3                 & 15 & 70   & 247   & 130           & 7.57  & 130            & 8.58  & 1             \\
			csum\_mux\_9                 & 30 & 140  & 941   & 941           & 10.49 & \textbf{285}   & 26.13 & \textbf{0.3}  \\
			gf2\textasciicircum{}4\_mult & 12 & 65   & 282   & 282           & 1.79  & 282            & 3.06  & 1             \\
			gf2\textasciicircum{}5\_mult & 15 & 97   & 1097  & 1069          & 6.64  & 1069           & 8.16  & 1             \\
			gf2\textasciicircum{}6\_mult & 18 & 135  & 4176  & 4176          & 5.18  & 4176           & 11.25 & 1             \\
			gf2\textasciicircum{}7\_mult & 21 & 179  & 16531 & 16531         & 38.26 & 16531          & 44.96 & 1             \\
			hwb8                         & 12 & 6446 & 2892  & \textbf{2521} & 8.16  & 2540           & 13.17 & 1.01          \\
			ham15-high                   & 20 & 1798 & 17512 & 11732         & 25.26 & \textbf{10092} & 56.43 & \textbf{0.86} \\
			ham15-low                    & 17 & 213  & 9612  & 3668          & 26.6  & \textbf{2116}  & 27.42 & \textbf{0.58} \\
			ham15-med                    & 17 & 452  & 10388 & 8292          & 17.81 & \textbf{3594}  & 42.32 & \textbf{0.43} \\			
			multiplier\_n25              & 25 & 203  & 1373  & 639           & 12.42 & \textbf{467}   & 13.9  & \textbf{0.73} \\
			qaoa\_n6                     & 6  & 270  & 886   & 766           & 1.7   & \textbf{634}   & 4.1   & \textbf{0.83} \\
			qcla\_adder\_10              & 36 & 181  & 369   & 175           & 15.74 & 175            & 23.9  & 1             \\
			qcla\_mod\_7                 & 26 & 294  & 1386  & \textbf{549}  & 13.26 & 641            & 22.49 & 1.17          \\
			qf21\_n15                    & 15 & 76   & 1035  & 1034          & 2.32  & \textbf{899}   & 5.76  & \textbf{0.87} \\
			tof\_10                      & 19 & 85   & 1545  & \textbf{55}   & 5.36  & 61             & 14.6  & 1.11          \\
			vqe\_uccsd\_n6               & 6  & 2282 & 1366  & 1366          & 1.85  & \textbf{684}   & 3.39  & \textbf{0.5}  \\ \hline
			gf2\textasciicircum{}8\_mult\_qc  & 24 & 115    & 217   & 60           & 9.79   & \textbf{59}   & 24.97  & \textbf{0.98} \\
			gf2\textasciicircum{}9\_mult\_qc  & 27 & 123    & 271   & 73           & 12.89  & \textbf{67}   & 21.26  & \textbf{0.92} \\
			gf2\textasciicircum{}10\_mult\_qc & 30 & 147    & 586   & \textbf{345} & 19.76  & 359           & 43.56  & 1.04          \\
			ham15-high\_qc                    & 20 & 1096   & 3743  & 2727         & 9.71   & \textbf{2551} & 31.24  & \textbf{0.94} \\
			ham15-med\_qc                     & 17 & 288    & 6700  & 2020         & 13.52  & \textbf{1508} & 34.91  & \textbf{0.75} \\
			hwb8\_qc                     & 12 & 4764 & 352   & \textbf{333}  & 3.46  & 335            & 20.72 & 1.01          \\
			hwb10\_qc                         & 16 & 23210  & 879   & 820          & 17.93  & \textbf{662}  & 40.34  & \textbf{0.81} \\
			hwb11\_qc                         & 15 & 63733  & 2783  & 2754         & 52.87  & 2742          & 71.91  & 1             \\
			hwb12\_qc                         & 20 & 122492 & 5629  & 5531         & 205.17 & \textbf{5486} & 255.73 & \textbf{0.99} \\
			mod\_adder\_1024\_qc              & 28 & 865    & 46770 & 177          & 18.05  & 177           & 44.98  & 1             \\
			qcla\_adder\_10\_qc               & 36 & 113    & 217   & \textbf{56}  & 17.12  & 64            & 39.65  & 1.14          \\
			qcla\_mod\_7\_qc                  & 26 & 176    & 196   & 51           & 5.2    & \textbf{49}   & 15.23  & \textbf{0.96} \\ \hline
			grcs\_16\_100\_0 & 16 & 100 & 8542   & 656           & 9.44    & \textbf{604}   & 11.75   & \textbf{0.92} \\
			grcs\_16\_100\_2 & 16 & 100 & 186098 & 19874         & 269.7   & \textbf{18850} & 1155.74 & \textbf{0.95} \\
			grcs\_16\_100\_4 & 16 & 100 & 6034   & 1198          & 13.31   & 1198           & 22.6    & 1             \\
			grcs\_16\_100\_6 & 16 & 100 & 23922  & 3032          & 47.88   & \textbf{2776}  & 97.89   & \textbf{0.92} \\
			grcs\_16\_100\_7 & 16 & 100 & 150298 & 17922         & 210.72  & \textbf{16930} & 543.7   & \textbf{0.94} \\
			grcs\_16\_100\_8 & 16 & 100 & 62962  & 6150          & 110.2   & \textbf{5702}  & 256.36  & \textbf{0.93} \\
			grcs\_20\_100\_0 & 20 & 100 & 2692   & 238           & 20.04   & \textbf{222}   & 35.79   & \textbf{0.93} \\
			grcs\_20\_100\_2 & 20 & 100 & 4956   & 2324          & 10.05   & \textbf{1364}  & 41.99   & \textbf{0.59} \\
			grcs\_20\_100\_7 & 20 & 100 & 1588   & 536           & 2.38    & \textbf{480}   & 8.01    & \textbf{0.9}  \\
			grcs\_20\_100\_8 & 20 & 100 & 9436   & 2376          & 29.03   & \textbf{2204}  & 74.75   & \textbf{0.93} \\
			grcs\_20\_110\_3 & 20 & 110 & 19398  & 2122          & 46.41   & \textbf{2106}  & 77.43   & \textbf{0.99} \\
			grcs\_20\_110\_5 & 20 & 110 & 19526  & 1818          & 26.19   & 1818           & 62.76   & 1             \\
			grcs\_20\_110\_7 & 20 & 110 & 6854   & 778           & 6.34    & 778            & 48.14   & 1             \\
			grcs\_20\_110\_8 & 20 & 110 & 69038  & 7418          & 83.66   & \textbf{7070}  & 194.31  & \textbf{0.95} \\
			grcs\_20\_110\_9 & 20 & 110 & 8262   & \textbf{1330} & 33.04   & 1378           & 58.48   & 1.04          \\
			grcs\_25\_100\_1 & 25 & 100 & 655    & \textbf{199}  & 6.28    & 203            & 12.18   & 1.02          \\
			grcs\_25\_100\_3 & 25 & 100 & 583    & 163           & 4.89    & \textbf{159}   & 10.41   & \textbf{0.98} \\
			grcs\_25\_100\_6 & 25 & 100 & 587    & 161           & 6.2     & 161            & 15.88   & 1             \\
			grcs\_25\_100\_7 & 25 & 100 & 583    & \textbf{163}  & 5.96    & 165            & 21.55   & 1.01          \\
			grcs\_25\_110\_0 & 25 & 110 & 1192   & 422           & 11.63   & 424            & 43.89   & 1             \\
			grcs\_25\_110\_1 & 25 & 110 & 1648   & \textbf{530}  & 22.44   & 554            & 51.03   & 1.05          \\
			grcs\_25\_110\_2 & 25 & 110 & 968    & \textbf{198}  & 14.69   & 206            & 29.61   & 1.04          \\
			grcs\_25\_110\_4 & 25 & 110 & 1648   & \textbf{304}  & 21.15   & 320            & 41.05   & 1.05          \\
			grcs\_25\_110\_6 & 25 & 110 & 1456   & \textbf{408}  & 19.89   & 420            & 40.32   & 1.03          \\
			grcs\_25\_110\_9 & 25 & 110 & 1520   & \textbf{300}  & 7.76    & 312            & 21.68   & 1.04          \\
			grcs\_25\_120\_2 & 25 & 120 & 1876   & 744           & 15.19   & \textbf{648}   & 31.64   & \textbf{0.87} \\
			grcs\_25\_120\_4 & 25 & 120 & 6988   & 1956          & 137.34  & \textbf{1933}  & 156.99  & \textbf{0.99} \\
			grcs\_25\_120\_5 & 25 & 120 & 1876   & 744           & 15.02   & \textbf{564}   & 26.48   & \textbf{0.76} \\
			grcs\_25\_120\_8 & 25 & 120 & 2844   & 1636          & 14.5    & \textbf{1508}  & 50.94   & \textbf{0.92} \\
			grcs\_25\_130\_0 & 25 & 130 & 24816  & 2232          & 1057.36 & \textbf{1490}  & 1224.3  & \textbf{0.67} \\
			grcs\_25\_130\_2 & 25 & 130 & 14148  & 1504          & 454.68  & \textbf{970}   & 481.11  & \textbf{0.64} \\
			grcs\_25\_130\_5 & 25 & 130 & 14148  & 1520          & 598.01  & \textbf{984}   & 728.27  & \textbf{0.65} \\ \hline
			Total            &    &     & 800251 & 149907 & 3925.01 & \textbf{133014} & 6655.69 & \textbf{0.89} \\ \hline
			\end{tabular}
	\end{table*}
	\end{center} 

\section{Conclusions}
In this paper, we integrate linear transformations into a recently proposed form of quantum computing: QMDDs.
We firstly show how linear transformations rearrange the entry of the original unitary matrix.
Then, we propose a compact and canonical representation of quantum computing: linearly transformed QMDDs (LTQMDDs).
Additionally, we design the linear sifting algorithm for LTQMDDs, obtaining more compact LTQMDDs.
Our experimental results justify that LTQMDDs are more compact than QMDDs.

\bibliographystyle{ACM-Reference-Format}
\bibliography{ICCAD-2022}

\end{document}